\documentclass[runningheads]{llncs}
\usepackage{color}
\usepackage{graphicx}
\usepackage{amssymb}
\usepackage{amsmath}
\usepackage{graphicx}
\usepackage[font=small, format=hang, labelfont=bf]{caption}
\usepackage[subrefformat=parens, labelfont=default]{subfig}
\usepackage{hyperref}
\usepackage{xspace}
\usepackage{wrapfig}
\usepackage{enumitem}

\title{Drawing Graphs with Vertices at Specified Positions
  and Crossings at Large Angles}

\author{%
Martin~Fink \inst1 \and Jan-Henrik~Haunert \inst1 \and Tamara~Mchedlidze
\inst2 \and Joachim~Spoerhase \inst1 \and Alexander~Wolff \inst1}

\authorrunning{M.~Fink et al.}
\titlerunning{Vertices at Specified Positions and Crossings at Large Angles}

\institute{Lehrstuhl f\"{u}r Informatik I, Universit\"{a}t
W\"{u}rzburg, Germany. \\
\texttt{http://www1.informatik.uni-wuerzburg.de}\\
\and Department of Mathematics, National Technical University of
Athens, Greece. \\ \texttt{mchet@math.ntua.gr}}

\graphicspath{{images/}}

\DeclareMathOperator{\arccot}{arccot}

\let\doendproof\endproof
\renewcommand\endproof{~\hfill\qed\doendproof}

\begin{document}

\maketitle

\begin{abstract}
  Point-set embeddings and large-angle crossings are two areas of
  graph drawing that independently have received a lot of attention in
  the past few years.  In this paper, we consider problems in the
  intersection of these two areas.  Given the point-set-embedding
  scenario, we are interested in how much we gain in terms of
  computational complexity, curve complexity, and generality if we
  allow large-angle crossings as compared to the planar case.  \par %
  We investigate two drawing styles where only bends or both bends and
  edges must be drawn on an underlying grid.  We present various results
  for drawings with one, two, and three bends per edge.
\end{abstract}

\section{Introduction}

In point-set-embeddability problems one is given not just a graph that
is to be drawn, but also a set of points in the plane that specify
where the vertices of the graph can be placed.  The problem class was
introduced by Gritzmann et al.~\cite{GritzmannMPP91} twenty years ago.
They showed that any $n$-vertex outerplanar graph can be embedded on
any set of $n$ points in the plane (in general position) such that
edges are represented by straight-line segments connecting the
respective points and no two edge representations cross.  Later on,
the point-set-embeddability question was also raised for other drawing
styles, for example, by Pach and Wenger~\cite{pw-epgfv-01} and by
Kaufmann and Wiese~\cite{jgaa/KaufmannW02} for drawings with polygonal
edges, so-called \emph{polyline drawings}.  In these and most other
works, however, planarity of the output drawing was an essential
requirement.

Recent experiments on the readability of drawings~\cite{HuangHE08}
showed that polyline drawings with angles at edge crossings close to
$90^\circ$ and a small number of bends per edge are just as readable
as planar drawings.  Motivated by these findings, Didimo et
al.~\cite{DidimoEL09} recently defined \emph{RAC} drawings where pairs
of crossing edges must form a right angle and, more generally,
$\alpha$\emph{AC} drawings (for $\alpha \in (0,90^\circ]$) where the
crossing angle must be at least~$\alpha$.  As usual, edges may not
overlap and may not go through vertices.

In this paper, we investigate the intersection of the two areas,
point-set embeddability (PSE) and RAC/$\alpha$AC.  Specifically, we
consider the following problems.

\medskip
\noindent
\emph{Problems RAC PSE and $\alpha$AC PSE.}  Given an $n$-vertex graph
$G=(V,E)$ and a set~$S$ of $n$ points in the plane, determine whether
there exists a bijection~$\mu$ between~$V$ and~$S$, and a polyline drawing
of~$G$ so that each vertex~$v$ is mapped to $\mu(v)$ and the drawing
is RAC (or $\alpha$AC).  If such a drawing 
exists and the largest number of bends per edge in the drawing is~$b$,
we say that $G$ \emph{admits a} RAC$_b$ (or an $\alpha$AC$_b$)
\emph{embedding on}~$S$.  \medskip

If we insist on straight-line edges, the drawing is completely
determined once we have fixed a bijection between vertex and point
set.  If we allow bends, however, PSE is also interesting \emph{with
mapping}, that is, if we are given a bijection~$\mu$ between vertex
and point set.  We call an embedding using $\mu$ as the mapping
\emph{$\mu$-respecting}.  The maximum number of bends over all edges
in a polyline drawing is the \emph{curve complexity} of the drawing.

We now list three results that motivate the study of RAC and $\alpha$AC
point-set embeddings---even for planar graphs.
\begin{itemize}
\item Rendl and Woeginger~\cite{RendlW93} have already
  considered a special case of the question we investigate in this
  paper, that is, the interplay between planarity and RAC in PSE.
  They showed that, given a set~$S$ of $n$ points in the plane, one
  can test in $O(n \log n)$ time whether a perfect matching admits a
  RAC$_0$ embedding on~$S$.  They required that edges are drawn as
  axis-aligned line segments.  They also showed that if one
  additionally insists on planarity, the problem becomes NP-hard.
\item Pach and Wenger~\cite{pw-epgfv-01} showed for the polyline drawing
  scenario with mapping that, if one insists on planarity, $\Omega(n)$
  bends per edge are sometimes necessary even for the class of paths
  and for points in convex position.
\item Cabello~\cite{Cabello06} proved that deciding whether a graph
  admits a planar straight-line embedding on a given point set is
  NP-hard even for $2$-outerplanar graphs.
\end{itemize}

In this paper, we concentrate on RAC PSE.
In order to measure the size of our drawings, we assume that the given
point set~$S$ lies on a grid of size $n \times n$ where $n=|S|$.  We further
assume that the points in $S$ are in \emph{general position}, that is, no two
points lie on the same horizontal or vertical line.  We call~$S$
an $n \times n$ \emph{grid point set}.  We require that,
in our output drawings, bends lie on grid points.  We concentrate on
two variants of the problem.  We either restrict the edges, which are
drawn as polygonal lines, to grid lines or we don't.  We refer to the
restricted version of the problem as \emph{restricted} RAC/$\alpha$AC PSE.
We treat the restricted version in Section~\ref{sec:restrictedRAC} and
the unrestricted version in Section~\ref{sec:unrestricted}.  The graphs we
study are always undirected.

Our results concerning restricted RAC and $\alpha$AC PSE are as follows.
\begin{itemize}
\item Every $n$-vertex binary tree admits a restricted RAC$_1$
  embedding on any $n \times n$ grid point set
  (Theorem~\ref{theorem:bin_tree}).  This is not known for the planar
  case---see our list of open problems in Section~\ref{sec:open}.  We
  slightly extend this result to graphs of maximum degree~3 that arise
  when replacing the vertices of a binary tree by cycles.  In the case
  of a single cycle, the statement even holds if the mapping is
  prescribed.  This is not true in the planar case: take the
  4-vertex cycle and the four points $(2,2),(4,4),(1,1),(3,3)$, in
  this order.

\item Given a graph, a point set on the grid, and a mapping~$\mu$, we
  can test in linear time whether the graph admits a $\mu$-respecting
  restricted RAC$_1$ point-set embedding
  (Theorem~\ref{theorem:2_3_4_prop}).  The same simple 2-SAT based
  test works in the planar case but of course fails more often.

\item Every $n$-vertex graph of maximum degree~3 admits a restricted
  RAC$_2$ embedding on any $n \times n$ grid point set even if the
  mapping is prescribed (Theorem~\ref{theorem:maxdeg3rrac2}).  Given a
  matching with $n$ vertices, a set of $n$ points on the $y$-axis, and
  a mapping~$\mu$, we can compute, in $O(n^2)$ time, a
  $\mu$-respecting restricted RAC$_{2}$ embedding of minimum area (to
  the right of the $y$-axis, see
  Theorem~\ref{theorem:rac2-one-side-min-area}).
\end{itemize}
Concerning unrestricted RAC and $\alpha$AC PSE, we show the following
results which all hold even if the mapping is prescribed.
\begin{itemize}
\item Every graph with $n$ vertices and $m$ edges admits a RAC$_3$
  embedding on any $n \times n$ grid point set within area
  $O\!\left((n+m)^2\right)$ (Theorem~\ref{theorem:K_n}).
  To RAC draw arbitrary graphs, curve complexity~3 is needed---even
  without 
  PSE~\cite{arikushietal2010}.
  In the planar case (with mapping), the curve complexity for PSE is
  $\Omega(n)$~\cite{pw-epgfv-01}. 
\item For any $\varepsilon > 0$, we get a $(\pi/2-\varepsilon)$AC$_2$
  drawing within area $O(nm)$ (Theorem~\ref{theorem:LAC_2bends}).  On a
  grid refined by a factor of $O(1/\varepsilon^2)$, we get a
  $(\pi/2-\varepsilon)$AC$_1$ drawing within area $O(n^2)$
  (Theorem~\ref{theorem:LAC_1bend}), which is optimal~\cite{GiacomoDLM10}.
  In the planar case, it is NP-hard to decide the existence of a
  1-bend point-set embedding---both with~\cite{gkossw-upg-09} and 
  without~\cite{jgaa/KaufmannW02} prescribed mapping.
\end{itemize}

\paragraph{Related work.}


Besides the above-mentioned work of Rendl and
Woeginger~\cite{RendlW93}, the study of PSE has primarily focussed on
the planar case, in connection with the drawing conventions
straight-line and polyline.  
A special case of the polyline drawings are \emph{Manhattan-geodesic}
drawings which require that the edges are drawn as monotone chains of
axis-parallel line segments.  This convention was recently introduced
by Katz et al.~\cite{KatzKRW10}.  They proved that Manhattan-geodesic
PSE is NP-hard (even for subdivisions of cubic graphs).  On the other
hand, they provided an $O(n \log n)$ decision algorithm for the
$n$-vertex cycle.  They also showed that Manhattan-geodesic PSE with
mapping is NP-hard even for perfect matchings---if edges are
restricted to the grid.  

Although RAC and $\alpha$AC drawings have been introduced very
recently, there is already a large body of literature on the problem.
Regarding the area of RAC drawings, Didimo et al.~\cite{DidimoEL09}
proved that an unrestricted RAC$_3$ drawing of an $n$-vertex graph
uses area $\Omega(n^2) \cap O(m^2)$.  Di Giacomo et
al.~\cite{GiacomoDLM10} showed that, for RAC$_4$ drawings, area
$O(n^3)$ suffices and that, for any $\varepsilon>0$, every $n$-vertex
graph admits a 
$(\pi/2-\varepsilon)$AC$_1$ drawing within area $\Theta(n^2)$.  Our
results for RAC$_3$ and AC$_1$ drawings (in Theorems~\ref{theorem:K_n}
and~\ref{theorem:LAC_1bend}) match the ones cited here, in spite of
the fact that vertex positions are prescribed in our case.


\section{Restricted RAC Point-Set Embeddings}
\label{sec:restrictedRAC}

In this section, we study restricted RAC point-set embeddings. It is
clear that only graphs with maximum degree $4$ may admit a restricted
RAC embedding on a point set. We start with the study of RAC$_1$
drawings.

\subsection{Restricted RAC$_1$ point-set embeddings}

The following result was independently achieved by Di Giacomo et
al.~\cite{DiGiacomoFFGK11}.

\begin{theorem}
  \label{theorem:bin_tree}
  Every binary tree has a restricted RAC$_1$ embedding on every $n
  \times n$ grid point set.
\end{theorem}

\begin{proof}
  Let $S$ be an $n \times n$ grid point set, let $T$ be a binary tree
  rooted at an arbitrary vertex~$r$, and let $v_1,\dots,v_n$ be a
  numbering of the vertices of $T$ given by a breadth-first-search
  traversal starting from $r$, i.e., $v_1=r$.  For $i=1,\dots,n$, let
  $T_i$ be the subtree of~$T$ rooted at vertex~$v_i$.

  Let~$p_1$ be the point in~$S$ such that the vertical line~$\ell_1$
  through~$p_1$ splits~$S_1=S$ according to~$T_1=T$, that is, we
  split~$S_1$ into a set~$S_2$ of $|T_2|$ points on its left and a
  set~$S_3$ of $|T_3|$ points on its right; see
  Fig.~\subref*{sfg:partition}.  Then we recursively pick points~$p_2$
  and~$p_3$ and lines~$\ell_2$ and $\ell_3$ that partition~$S_2$
  and~$S_3$ according to~$|T_2|$ and~$|T_3|$.  We continue
  until we arrive at the leaves of~$T$.  This process determines
  points~$p_1,\dots,p_n$ and lines~$\ell_1,\dots,\ell_n$ such that for
  $i=1,\dots,n$ point~$p_i$ lies on~$\ell_i$.  We simply map
  vertex~$v_i$ to point $p_i$ for $i=1,\dots,n$.

  Consider an index $i \in \{1,\dots,n\}$.  Our mapping makes sure
  that one subtree of~$T_i$ is drawn on the left of $\ell_i$ and the
  other on the right of $\ell_i$.  Let $v_j$ and $v_{j+1}$ be the
  children of $v_i$.  We draw the edges $(v_i,v_j)$ and
  $(v_i,v_{j+1})$ such that their horizontal segments are both
  incident to~$v_i$, see Fig.~\subref*{sfg:edgedrawing}.

  The resulting drawing is clearly a RAC drawing since all edges are
  restricted to the grid.  Since~$S$ is in general position, no two
  edges can overlap except if they are incident to the same vertex.
  If we direct the edges of $T$ away from the root, then, by our
  drawing rule, in any vertex $v_i$ of~$T$ the incoming edge arrives
  in~$p_i$ with a vertical segment and the outgoing edges leave~$p_i$
  with horizontal segments in opposite directions.
\end{proof}


We can, of course, also find a restricted RAC$_1$ embedding for paths
as special binary trees. Actually, we can embed every $n$-vertex path or cycle
on any $n \times n$ grid point set, even with mapping: we simply leave
each point horizontally and enter the next one vertically in the order
prescribed by the mapping.

It would, of course, be nice to generalize these embeddability results
for binary trees and cycles (without given mapping) to larger classes of
graphs, e.g., outerplanar graphs of maximum degree 3. This seems,
however, to be quite difficult.  A class of graphs
that we can embed are maxdeg-3 cactus graphs that are
constructed from binary trees by replacing vertices by cycles.

We can embed graphs of this type on any $n \times n$ grid point set by
adjusting the embedding algorithm for binary trees. The basic idea is
to treat each cycle similarly to a single tree vertex.  We do
this by reserving the adequate number of consecutive columns for the
vertices of the cycle in the middle of the drawing area for the current
subtree.  We connect the cycle to the left subtree by leaving the
leftmost reserved point to the left.  We deal with the right subtree
symmetrically.  One of the points reserved for the cycle---say,
$z$---must be connected to the parent vertex (or cycle).  The
difficulty is to make a cycle from the reserved points in such a way
that~$z$ can be entered \emph{vertically} from its parent, which has
been embedded before.  This is possible but the proof is technical,
and, hence, left for the appendix.  Summing up, we get the following
result.

\newcommand{\binarycactus}{%
  Let $G$ be an $n$-vertex graph of maximum degree~3 that arises
  when replacing the vertices of a binary tree by cycles
  and let $S$ be an $n \times n$ grid point set.
  Then $G$ admits a restricted RAC$_1$ embedding on $S$.
}
\newcounter{binarycactustheorem}
\setcounter{binarycactustheorem}{\value{theorem}}

\begin{theorem}
  \label{theorem:binary_cycle_tree}
  \binarycactus
\end{theorem}

In the proofs of the previous theorems we exploited the fact
that we could choose the vertex--point mapping as needed.
Figure~\ref{fig:counterexample} shows a 6-vertex binary tree that does
\emph{not} have a restricted RAC$_1$ drawing on the given point set if
the vertex--point mapping is fixed as indicated by the edges.  Hence, we turn to
the corresponding decision problem.  We characterize situations when a
restricted RAC$_1$ point-set embedding with mapping exists.

\begin{figure}[tb]
  \begin{minipage}[b]{.68\linewidth}
    \subfloat[Partition of the point set\label{sfg:partition}]%
    {$\;$\includegraphics{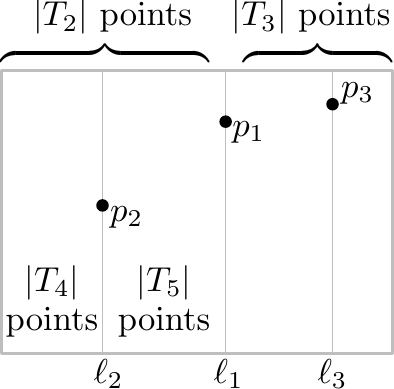}$\;$}
    \hfill
    \subfloat[Drawing of the edges\label{sfg:edgedrawing}]%
    {$\;\;$\includegraphics{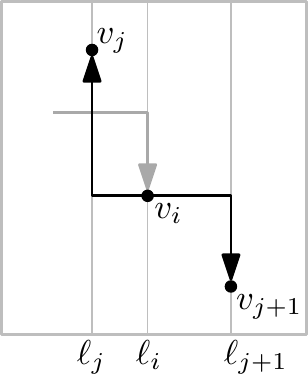}$\;\;$}

    \caption{Illustrations for the proof of Theorem~\ref{theorem:bin_tree}.}
    \label{fig:bin_tree_proof}
  \end{minipage}
  \hfill
  \begin{minipage}[b]{.3\linewidth}
    \centering
    \includegraphics{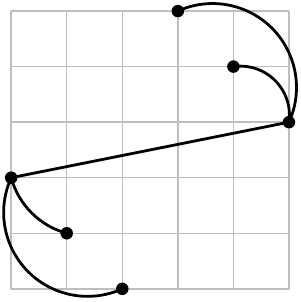}
    \caption{A binary tree without restricted RAC$_1$ drawing.}
    \label{fig:counterexample}
  \end{minipage}
\end{figure}

\begin{theorem}
  \label{theorem:2_3_4_prop}
  Let $G$ be an $n$-vertex graph of maximum degree~$4$, let~$S$ be
  an $n \times n$ grid point set, and let $\mu$ be a vertex--point
  mapping.  We can test in $O(n)$ time whether $G$ admits a
  $\mu$-respecting restricted RAC$_1$ embedding on~$S$ and, if yes,
  construct such an embedding within the same time bound.
\end{theorem}

\begin{proof}
  We use a 2-SAT encoding to solve the problem.  A similar approach
  was used by Raghavan et al.~\cite{RaghavanCS86} to deal with the
  planar case.  We associate each edge~$uv$ of~$G$ with a Boolean
  variable~$x_{uv}$.  The two possible drawings of edge~$uv$
  correspond to the two literals $x_{uv}$ and $\neg x_{uv}$.
  Due to the fact that $S$ is in general position, only drawings
  of edges incident to the same vertex can possibly overlap.

  Now we construct a 2-SAT formula~$\phi$ as follows.  Consider a pair
  of drawings of edges $uv$ and~$uw$ that overlap.  Assume that
  $x_{uv}$ and $\neg x_{vw}$ are the literals corresponding to the two
  edge drawings.  Then we add the clause $\neg (x_{uv} \wedge 
  \neg x_{uw}) = \neg x_{uv} \vee x_{uw}$ to~$\phi$.

  It is clear that~$\phi$ is satisfiable if and only if $G$ has a
  $\mu$-respecting RAC$_1$ embedding on~$S$ without overlapping edges. 
  Recall that the maximum degree of~$G$ is~4.  Hence, $\phi$ contains
  at most $n \cdot \binom{4}{2} \cdot 4$ clauses.  Since the
  satisfiability of a 2-SAT formula can be decided in time linear in
  the number of clauses~\cite{ets-ctmfp-76}, the testing can be done in
  $O(n)$ time.
\end{proof}

\subsection{Restricted RAC$_2$ point-set embeddings}
\label{sec:rac2-embeddings}

As in the previous subsection, it is clear that only graphs of maximum
degree~4 can be drawn with the grid restriction.
Consider, for a moment, a specialized restricted RAC$_2$
drawing convention that requires the first and the last (of the three)
segments of an edge to go in the same direction---a \emph{bracket}
drawing. If we do not restrict the drawing area, then the problem of
bracket embedding a graph~$G$ on an $n \times n$ grid point set is
equivalent to $4$-edge coloring~$G$.  The idea is that the four colors
encode the direction of the first and last edge segment (going up,
down, left, or right) and that the second edge segment is drawn
sufficiently far away.  The edge coloring makes sure that no two edges
incident to the same vertex overlap.
It is known that any graph of maximum degree~3 is
$4$-edge colorable and that such a coloring can be found in linear
time~\cite{skulrattanakulchai20024}.  Let us summarize.
\begin{theorem}
    \label{theorem:maxdeg3rrac2}
    Every graph $G$ of maximum degree 3 admits a restricted
    RAC$_2$ embedding on any $n \times n$ grid point set with any
    vertex--point mapping.
\end{theorem}

Note that there are graphs of maximum degree 4 that do not admit a
$4$-edge coloring, but do admit a restricted
RAC$_2$ embedding at least for some grid point sets (see
Figure~\ref{fig:k5_drawing_rac2} in the appendix for such an embedding
of $K_5$).

Now we turn to the problem of minimizing the drawing area. 
Observe that there are examples of a
graph~$G$, a grid point set~$S$, and a mapping~$\mu$ such that $G$
does not admit a restricted RAC$_2$ point-set embedding on~$S$ with
mapping~$\mu$ if we insist that the drawing lies within the bounding box
of~$S$, see Fig.~\ref{fig:rrac2-counterexample}. 

\begin{wrapfigure}{r}{0.2\textwidth}
  \centering
  \vspace{-3ex}
  \includegraphics[width=\linewidth]{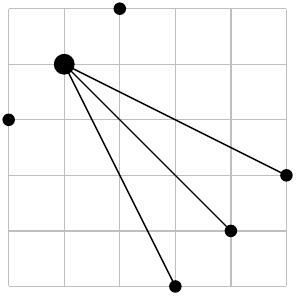}
  \caption{Counter\-example.}
  \label{fig:rrac2-counterexample}
  \vspace{-4ex}
\end{wrapfigure}

We conjecture that restricted RAC$_2$ PSE is NP-hard.  Therefore, we
consider the special case where $S$ is one-dimensional.  More
precisely, we are looking for a \emph{one-page RAC$_2$ book embedding}
with given mapping.  Recall that, generally, a $k$-page book embedding
asks for a mapping of the vertices to points on a line, the
\emph{spine} of the book, and a mapping of the edges to the pages of
the book (that is, half-planes incident to the spine) such that, for
each page, the edges on that page can be drawn without crossings.

Clearly, in this setting, each vertex can only have
degree~1, hence the given graph must be a (perfect) matching.  Given
these restrictions, we can minimize the area of the drawing.

\begin{theorem}
  \label{theorem:rac2-one-side-min-area}
  Let $S$ be a set of $n$ points on the $y$-axis, let $G$ be a
  matching consisting of $n/2$ edges, and let $\mu$ be a vertex--point
  mapping.  A minimum-area $\mu$-respecting restricted RAC$_2$ drawing
  of~$G$ to the right of the $y$-axis can be computed in $O(n^2)$ time.
\end{theorem}
\begin{proof}
  If $S$ contains pairs of neighboring points that correspond to edges
  of the given matching, we connect each of them by a (vertical)
  straight-line segment.  To draw any of the remaining edges of the
  matching in a restricted RAC$_2$ fashion, we must connect its
  endpoints by two horizontal segments leaving the $y$-axis to the
  right and a vertical segment that joins the horizontal segments.
  As~$G$ is a matching, only vertical segments can overlap.  In order
  to minimize the drawing area, we, thus, have to minimize the number
  of vertical lines, the \emph{layers}, needed to draw the vertical
  segments of all edges without overlap.

  Let $G' = (V', E')$ with
  $V' = E$ and an edge connecting each pair of edges of~$G$
  that cannot use the same layer.  Clearly, assigning the edges of~$G$
  to the minimum number of layers is the same as coloring the vertices
  of~$G'$ with the minimum number~$\chi'$ of colors.

  Graph~$G'$ is an interval graph: for edge~$uv$ of~$G$---a vertex
  of~$G'$---the interval is $[\mu(u),\mu(v)]$.  Hence, a coloring
  of~$G'$ using $\chi'$ colors can be computed in $O(|V'| + |E'|) =
  O(n^2)$ time~\cite{Olariu91}.
  This coloring yields an assignment of the edges to the minimum
  number of layers, which in turn corresponds to a minimum-area
  restricted RAC$_2$ drawing: we simply use the first~$\chi'$
  vertical grid lines immediately to the right of the $y$-axis for the
  layers of the vertical edge segments.
\end{proof}


If we are not given a prescribed mapping, then the problem becomes
easy for all graphs of maximum degree~2.  We simply draw the connected
components of $G$, which are paths or cycles, one after the other
using the points in~$S$ from top to bottom. This can be done using
only the $y$-axis for paths and using only one column right of the
$y$-axis for cycles.

If we abandon the restriction to draw edges on the grid and relax the
constraint on the crossing angle, we can find, for \emph{any} graph,
an $\alpha$AC$_2$ embedding on any point set on the $y$-axis with an
arbitrary mapping, see the comment after the proof of
Theorem~\ref{theorem:LAC_2bends}. 

\section{Unrestricted RAC and $\alpha$AC Point-Set Embeddings}
\label{sec:unrestricted}

Didimo et al.~\cite{DidimoEL09} have shown that any graph with $n$
vertices and $m$ edges admits a RAC$_3$-drawing within area $O(m^2)$.
Their proof uses an algorithm of Papakostas and
Tollis~\cite{PapakostasT00} for drawing graphs such that each vertex
is represented by an axis-aligned rectangle and each edge by an
\emph{L-shape}, that is, an axis-aligned 1-bend polyline.  Didimo et
al.\ turn such a drawing into a RAC$_3$-drawing by replacing each
rectangle with a point.  In order to make the edges terminate at these
points, they add at most two bends per edge.  We now show how to
compute a RAC$_3$-drawing of the same size (assuming $n \in
O(m)$)---although we are restricted to the given point set.

Note that curve complexity~3 is actually necessary for RAC drawing
arbitrary graphs---even without a prescribed point set:
Arikushi~et~al.~\cite{arikushietal2010} showed that RAC$_2$
drawings only exist for graphs with a linear number of edges.

\begin{theorem}
  \label{theorem:K_n}
  Let $G$ be a graph with $n$ vertices and $m$ edges and let
  $S$ be an $n \times n$ grid point set.
  Then~$G$ admits a RAC$_3$-drawing on~$S$ (with or without
  given vertex--point mapping) within area $O\left((n+m)^2\right)$.
\end{theorem}
\begin{proof}
    If the vertex--point mapping $\mu$ is not given, let $\mu$ be an
    arbitrary mapping.  Let $v_{1}, \ldots, v_{n}$ be an ordering
    of~$V$ so that $p_i := \mu(v_{i})$ has $x$-coordinate~$i$.  We
    construct a RAC$_3$-drawing as follows.  Each edge has---after
    insertion of ``virtual'' bends---exactly three bends and four
    straight-line segments.  We ensure that intersections involve only
    the ``middle'' segments of edges, and that these middle segments
    have only slope $+1$ or $-1$.

    For an edge $uv$, we call the bend directly connected to $u$ a
    \emph{$u$-bend}, the bend directly connected to~$v$ a
    \emph{$v$-bend}, and the remaining bend the \emph{middle bend}.
    We start constructing the drawing by
    placing the $v$-bends for each vertex $v$, starting with~$v_{n}$.
    We set the $y$-coordinate~$y_n$ of the first $v_n$-bend to~0.  
    Then, for $i=n,n-1,\dots,1$, observe that there are
    exactly $\deg v_i$ many $v_i$-bends, which we place in column $i+1$
    starting at $y$-coordinate $y_i$ below the $n \times n $ grid using
    positions $\{ (i+1, y_{i}), (i+1, y_{i} - 2), (i+1, y_{i} - 4),
    \ldots, (i+1, y_{i} - 2 \cdot (\deg v_i - 1) \}$, see
    Figure~\ref{fig:rac3-construction}. We
    connect each vertex with its associated bends without
    introducing any intersection since we stay inside the area
    between columns $i$ and $i+1$.  We set $y_{i-1} = y_{i} - 2 \cdot
    (\deg v_{i} - 1) - 3$.  If~$v_i$ has degree~0, we do not place
    bends but set $y_{j-1} = y_j - 3$ to avoid overlaps and
    crossings.  Then we continue with~$v_{i-1}$.

    \begin{figure}[tb]
      \begin{minipage}[b]{0.5\linewidth}
        \centering
        \includegraphics{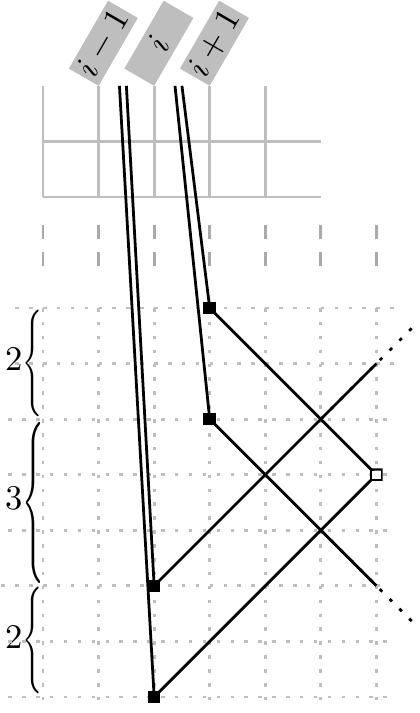}
        \caption{Construction of a RAC$_3$ drawing.}
        \label{fig:rac3-construction}
      \end{minipage}
      \hfill
      \begin{minipage}[b]{0.5\linewidth}
        \centering
        \includegraphics{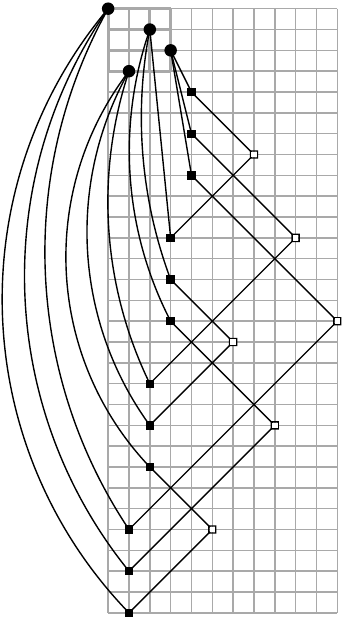}
        \caption{RAC$_3$-drawing of $K_4$ as in the proof of
        Theorem \ref{theorem:K_n}.  For the sake of clarity, we
        replaced some straight-line segments by circular arcs.}
        \label{fig:rac3-construction-example}
      \end{minipage}
    \end{figure}

    Since we place the bends from right to left and from top to bottom
    by moving our ``pointer'' by $L_1$- (or Manhattan) distances~2
    or~4, each pair of these bends has even Manhattan distance.  To
    draw an edge~$uv$, we first select a ``free'' $u$-bend position
    and a free $v$-bend position.  For the two middle segments, we use
    slopes $+1$ and $-1$ such that the middle bend is to the right of
    the $u$- and $v$-bend.  Since $u$- and $v$-bend have even
    Manhattan distance, the middle bend has integer coordinates.

    Let $u$ and $v$ be two vertices with $u$-bend $b_u$ and $v$-bend
    $b_v$, respectively. The segments $\overline{ub_u}$ and
    $\overline{vb_v}$ cannot intersect; we want to see that the middle
    segment starting at $b_u$ also cannot intersect
    $\overline{vb_v}$. Such an intersection can only occur if $u$ lies
    to the left of~$v$.
    By our construction, $b_v$ lies, in this case, above $b_u$ with a
    $y$-distance that is greater than their $x$-distance. As all
    middle segments have a slope of at most $+1$, $b_v$ lies above the
    relevant middle segment, which can, hence, not intersect
    $\overline{vb_v}$.

    It remains to show the space limitation. Clearly, the drawing of any
    edge requires not more horizontal than vertical space. On the other
    hand, for any vertex $v$, we need at most $2 \cdot \deg v + 3$
    rows below the grid, resulting in a total vertical space requirement
    of $O(n+m)$.  This completes the proof.
\end{proof}


In the remainder of this section we focus on $\alpha$AC point-set
embeddings.  We show that both area and curve complexity can be
significantly improved if we soften the restriction on the crossing angles.
Our results hold for both scenarios, with and without vertex--point mapping.

\begin{theorem}
    Let $G$ be a graph with $n$ vertices and $m$ edges, let $S$ be a
    $n \times n$ grid point set, and let $0 < \varepsilon <
    \frac{\pi}{2}$.  Then $G$ admits a $(\frac{\pi}{2} -
    \varepsilon)$AC$_2$ embedding on~$S$ (with or without given
    vertex--point mapping) within area $O(n(m + \cot \varepsilon)) =
    O(n(m + 1/\varepsilon^2))$.
    \label{theorem:LAC_2bends}
\end{theorem}
\begin{proof}
    If the vertex--point mapping $\mu$ is not given, let $\mu$ be
    arbitrary.  Let $v_{1}, \dots, v_{n}$ be an ordering of~$V$ so that
    $p_i:=\mu(v_{i})$ has $x$-coordinate $i$.  
    Each edge $e = uv$ has exactly two bends,
    a $u$-bend and a $v$-bend (with the obvious meanings).  
    For $i=1,\dots,n$, we place all $v_{i}$-bends in column $i+1$.  We
    make all middle segments of edges horizontal.  Thus, the bends for
    an edge $e = v_{i}v_{j}$ are at positions $(i+1,y)$ and $(j+1,y)$
    in some row $y < 0$ below the original grid, see
    Figure~\ref{fig:lac2-construction}.  By using a dedicated row for
    each edge, we achieve that no two middle segments
    intersect.  By construction, no two first or last edge segments
    intersect.  Hence, crossings occur only between the horizontal
    middle segments and first or last segments.  By
    making the $y$-coordinates of the middle segments small
    enough, we will achieve that all crossing angles are at least $\pi/2
    - \varepsilon$.

    \begin{figure}[tb]
      \begin{minipage}[t]{0.47\linewidth}
        \centering
        \includegraphics{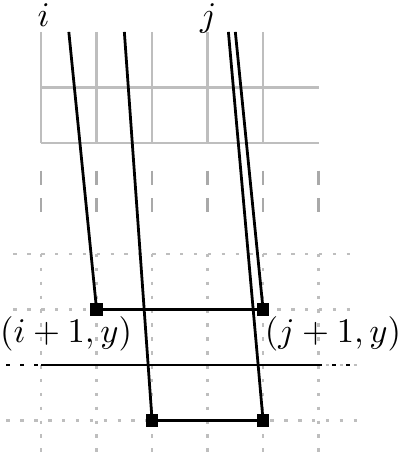}
        \caption{Constructing a 2-bend drawing with large crossing
          angles.}
        \label{fig:lac2-construction}
      \end{minipage}
      \hfill
      \begin{minipage}[t]{0.46\linewidth}
        \centering
        \includegraphics{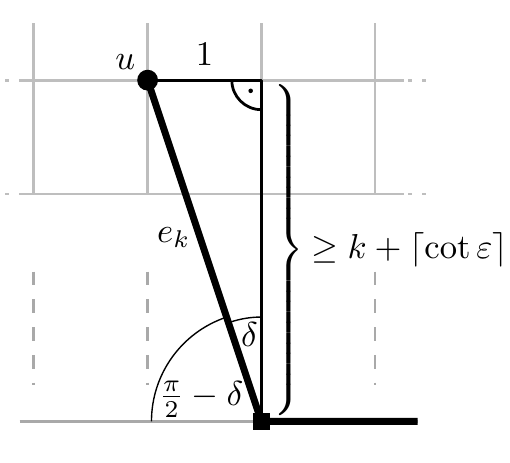}
        \caption{Angles in the 2-bend drawing.}
        \label{fig:lac2-angle}
      \end{minipage}
    \end{figure}

    Let $\{ e_{1}, \ldots, e_{m} \}$ be the set of edges of~$G$,
    and let $uv:=e_{k}$ be one of these edges.  We set the
    $y$-coordinates of the middle segment of $e_{k}$ to $- k -
    \lceil\cot \varepsilon\rceil$.  Let $e_{k'}$ be an edge whose
    horizontal segment intersects the first segment of~$e_{k}$. The
    crossing angle is $\pi/2 - \delta$, where $\delta$ is the angle
    between the vertical line through the $u$-bend and the first segment
    of~$uv$, see Figure~\ref{fig:lac2-angle}.  We have $\delta \le 
    \arccot (k + \lceil\cot \varepsilon\rceil) \le \varepsilon$.  Thus,
    the crossing angle is at least $\pi/2 - \varepsilon$.  Note that
    $\cot \varepsilon \in O(1/\varepsilon^2)$.
\end{proof}

We used only the fact that no two points lie in the same column.  Hence,
the statement of the theorem does not change if we allow the points
to lie on a single horizontal (or, by rotation, vertical) line as in
Section~\ref{sec:rac2-embeddings}.

In Theorem~\ref{theorem:LAC_2bends}, we required the bends to lie on
points of the given grid.  The following result shows that we need only one
bend per edge if we allow the bends to lie on points of a
\emph{refined} grid.  For fixed $\epsilon>0$, our new drawings need less
area than those of Theorem~\ref{theorem:LAC_2bends}; even in terms of
the refined grid.

\begin{theorem}
    Let $G$ be a graph with $n$ vertices, 
    let~$S$ be an $n \times n$ grid point set, and let $0 < \varepsilon <
    \frac{\pi}{2}$.  Then $G$ admits a $(\frac{\pi}{2} -
    \varepsilon)$AC$_1$ embedding on~$S$ (with or without given
    vertex--point mapping) on a grid that is finer than the original
    grid by a factor of $\lambda \in O(\cot \varepsilon) =
    O(1/\varepsilon^2)$.
    \label{theorem:LAC_1bend}
\end{theorem}
\begin{proof}
    If the mapping $\mu$ is not given, let $\mu$ be an arbitrary
    mapping.  The idea of our construction is as follows.  For each
    edge, we first choose 
    one of the two possible drawings on the grid lines with one bend.
    This gives us a drawing of the graph with many overlaps of
    edges.  Then, we slightly twist each edge such that its
    horizontal segment becomes \emph{almost horizontal}, meaning it
    gets a negative slope close to 0.  At the same time, we make the
    vertical segment \emph{almost vertical}, meaning it gets a very
    large positive slope, see Figure~\ref{fig:K_4-LAC_1}.

    As we want all bends to be on grid points, we first refine the
    grid by an integral factor of $\lambda = \lceil 1 + \cot
    \varepsilon \rceil$.  We do this by inserting, at equal distances,
    $\lambda-1$ new rows or columns between two consecutive grid rows
    or columns, respectively.  Now, a point $s=(a,b) \in S$ lies at
    $(\lambda a, \lambda b)$ w.r.t.\ the new $ \lambda n \times
    \lambda n$ grid.

    Let $e$ be an edge and let $(e_x, e_y)$ be the original position
    of the bend of~$e$ w.r.t.\ the new
    grid.  We choose the new position of the bend to be the unique grid
    point diagonally next to $(e_x, e_y)$ such that the horizontal and
    vertical segments of~$e$ become almost horizontal and almost vertical,
    respectively. If we apply this construction to all edges, we get a
    drawing in which none of the almost horizontal and almost vertical
    segments belonging to some vertex $v$ can overlap.  Moreover, two
    almost horizontal or two almost vertical segments
    belonging to different vertices neither overlap nor intersect
    due to~$S$ being in general position. Thus, each crossing involves an
    almost horizontal and an almost vertical segment.

    Let $e_{1}$ and $e_{2}$ be two crossing edges such that the almost
    horizontal segment involved in the crossing belongs to $e_1$. We can
    assume that the smaller angle of the crossing occurs to the top
    left of the crossing; the other case is symmetric by a rotation
    of the plane. Let $\delta^-$ be the angle formed by the almost
    horizontal segment of $e_1$ and a horizontal line, and let
    $\delta^+$ be the angle formed by the almost vertical segment of
    $e_1$ and a vertical line, see Figure~\ref{fig:lac1-angle}. Then the
    crossing angle of~$e_1$ and~$e_2$ is $\alpha = \pi/2 - \delta^{-}
    + \delta^{+} \ge \pi/2 - \delta^{-}$.  
    For $\delta^{-}$ to be maximal, the horizontal
    length~$l$ of the almost horizontal segment has to be minimal. As
    this length cannot be less than $\lambda-1$, we get $\delta^{+} \le \arccot
    (\lambda-1) \le \varepsilon$. Hence, the crossing angle $\alpha$ is at least
    $\pi/2 - \varepsilon$.
\end{proof}

    \begin{figure}[tb]
          \begin{minipage}[t]{0.37\linewidth}
            \centering
            \includegraphics[scale=1.1]{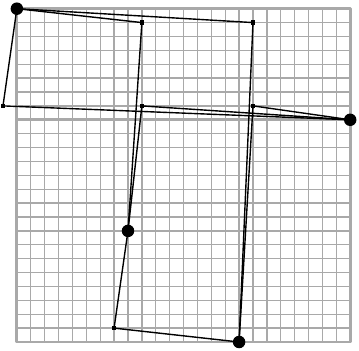}
            \caption{Drawing of $K_4$ on a grid refined by factor
              $\lambda = 8$.} 
            \label{fig:K_4-LAC_1}
          \end{minipage}
          \hfill
          \begin{minipage}[t]{0.61\linewidth}
            \centering
            \includegraphics{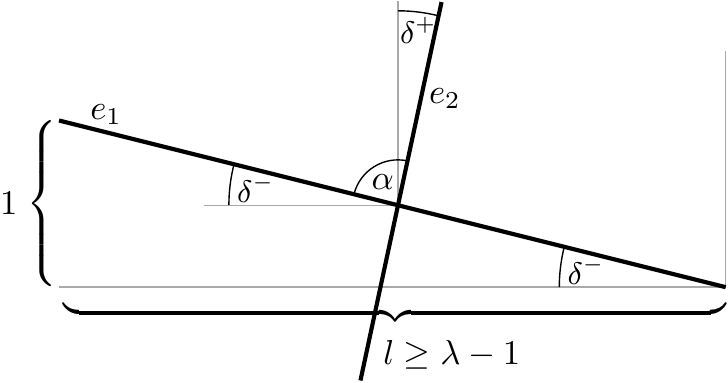}
            \caption{Angles in the 2-bend-drawing.}
            \label{fig:lac1-angle}
          \end{minipage}
    \end{figure}

Note that we leave the original grid by at most one row or column
of the refined grid in each
direction. Hence, the area requirement is $O( (n \cdot
\cot \varepsilon)^2)$ in terms of the finer grid.  We
argue that our area bounds are quite reasonable: for a minimum
crossing angle of~$70^\circ$, the drawings provided by
Theorems~\ref{theorem:LAC_2bends} and~\ref{theorem:LAC_1bend} use
grids of sizes at most $n(m+3)$ and $(3n)^2$, respectively.

\section{Open Problems}
\label{sec:open}

In this paper, we have opened an interesting new area: the
intersection of point-set embeddability and drawings with crossings at
large angles.  We have done a few first steps, but we leave open a
large number of questions.  We start with the restricted case where
vertices, bends, and edges must lie on the grid.
\begin{enumerate}
\item Does every $n$-node binary tree have a restricted \emph{planar}
  1-bend embedding on any $n\times n$ grid point set?
\item Does every $n$-node ternary tree have a restricted RAC$_1$
  embedding on any $n\times n$ grid point set? 
\item What about outerplanar graphs?
\item Can we efficiently test whether a given graph has a restricted
  RAC$_1$ embedding on a given $n\times n$ grid point set?
\item What about RAC$_2$?
\end{enumerate}
Recall that in the unrestricted case we don't require edges to lie
on the grid.
\begin{enumerate}[resume]
\item Can we efficiently test whether a given graph has a RAC$_2$
  embedding on a given $n\times n$ grid point set?  If yes, can we
  minimize its area?
\item Di Giacomo et al.~\cite{GiacomoDLM10} have shown that any graph
  with $n$ vertices and $m$ edges admits a RAC$_4$-drawing that uses
  area $O(n^3)$.  Can we achieve the same in our PSE setting?
\end{enumerate}

\medskip
\noindent
\textbf{Acknowledgments.}  We thank Beppe Liotta for suggesting the
idea behind Theorem~\ref{theorem:LAC_1bend} to us.

\bibliographystyle{abbrv}
\bibliography{abbrv,lncs,bibliography}

\newpage
\appendix

\section*{\appendixname}

\setcounter{theorem}{\value{binarycactustheorem}}
\begin{theorem}
  \binarycactus
\end{theorem}
\begin{proof}
  We adjust the embedding algorithm for binary trees to work with the
  new graph class. The basic idea is to treat each cycle similar to a
  single vertex of a binary tree. We do this by reserving the adequate
  number of consecutive columns for the nodes of the cycle in the
  middle of the drawing area for the current subtree when splitting
  into the drawing areas for the subtrees. The subtrees are connected
  to the cycle by leaving one point to the right, and one point to the
  left, respectively.  The most difficult part is to connect the
  reserved nodes to a cycle in such a way that the point representing
  the vertex that is the connector to the parent vertex (or cycle,
  respectively), which was embedded before, can be connected by
  entering the node with a vertical segment such that the connections
  to the left and the right are possible.

  Let $C$ with $ k := |C| \ge 3$ be the cycle representing the root of the
  current subtree with vertices $u$ and $v$ connecting the cycle to
  the roots $r_l$ and $r_r$ of its left and right subtrees,
  respectively, and a vertex $z$ connecting $C$ to its parent $r$.
  Let $S' = \left\{ p_{1}, \ldots,
  p_{k} \right\}$ be the set of points reserved for $C$ in consecutive
  columns ordered from left to right. The edge connecting $C$ to the
  left and right subtree enter the points representing $u$ and $v$
  from left and right, respectively, while the edge connecting $z$ to
  $r$ enters $z$ from above or below, depending on the
  $y$-coordinate of the point chosen to represent $z$. Let $y_r$ be
  the $y$-coordinate of $r$. We analyze the different cases.

  \begin{enumerate}
    \item
      \label{case:many_vertices}
      Vertex $z$ has a neighbor $w \neq u,v$ in $C$ and
      $k \ge 5$:

      Set $\mu(u) = p_1$ and $\mu(v) = p_k$. Either above or below the
      line $y = y_{r}$ we find two points $p, p' \in S' \setminus
      \left\{ p_1, p_k \right\}$. Let $p$ be the one closer to
      the line $y = y_{r}$. We set $\mu(p) = z, \mu(p') = w$ and
      draw the edge $wz$ such that $p$ is entered vertically. Then we
      can complete the cycle such that each point is incident to a
      horizontal and a vertical segment, see
      Figure~\ref{fig:cactusbintree_case1}. It is easy to see that the
      connections to $r, r_l$ and $r_r$ can now be drawn without
      overlap.

      \begin{figure}[tbh]
        \begin{minipage}[b]{0.48\linewidth}
					\centering
						\includegraphics{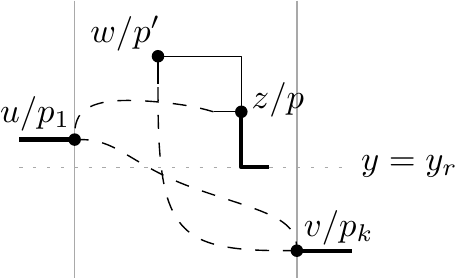}
					\caption{Drawing of $C$ with at least 5 vertices and a
					neighbor $w$ of $z$.}
					\label{fig:cactusbintree_case1}
        \end{minipage}
        \hfill
        \begin{minipage}[b]{0.48\linewidth}
          \centering
            \includegraphics{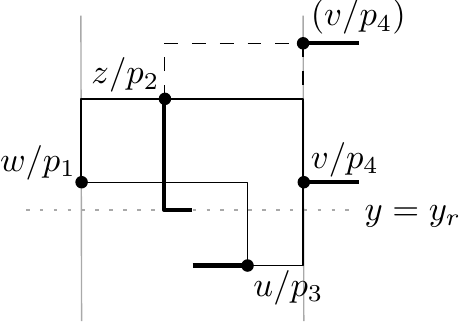}
          \caption{Drawing of $C$ with $k=4$ and $p_4$ above $p_3$.}
          \label{fig:cactusbintree_case21}
        \end{minipage}
      \end{figure}

    \item Vertex $z$ has a neighbor $w \neq u,v$ in $C$ and
      $k = 4$:

      Let $C = \left( u, w, z, v \right)$ the other case being
      symmetric. If $p_2$ and $p_3$ both lie either below or above $y
      = y_r$ we can proceed as in case \ref{case:many_vertices}.
      If $p_2$ lies above $r$ and $p_3$ below we have two subcases
      depending on where $p_4$ is:

      \begin{itemize}
        \item $p_4$ lies above $p_3$: We can draw $C$ as shown in
          Figure~\ref{fig:cactusbintree_case21}.

        \item $p_4$ lies below $p_3$: We can draw $C$ as shown in
          Figure~\ref{fig:cactusbintree_case22}.
      \end{itemize}

      \begin{figure}[tbh]
        \begin{minipage}[b]{0.48\linewidth}
          \centering
            \includegraphics{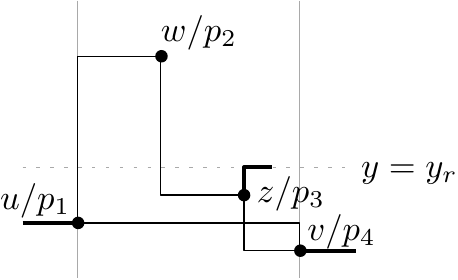}
          \caption{Drawing of $C$ with $k=4$ and $p_4$ below $p_3$.}
          \label{fig:cactusbintree_case22}
        \end{minipage}
				\hfill
				\begin{minipage}[b]{0.48\linewidth}
					\centering
					\includegraphics{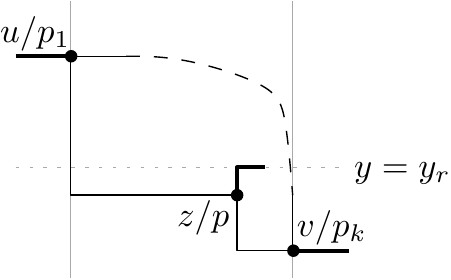}
					\caption{Drawing of $C$ with $u,v$ as neighbors of
					$z$ and one point vertically between $p_1$ and $p_k$.}
					\label{fig:cactusbintree_case31}
				\end{minipage}
      \end{figure}

      If $p_3$ is above $r$ and $p_2$ below the cases are symmetric.

    \item The two neighbors of $z$ are $u$ and $v$.  If there is one point
      $p \in S' \setminus \left\{ p_1, p_k \right\}$ that is
      vertically between $p_1$ and $p_k$, then we set $\mu(u) = p_1, \mu(v)
      = p_k$ and $\mu(z) = p$ and draw $C$ as in
      Figure~\ref{fig:cactusbintree_case31}, where the second path
      connecting $u$ and $v$ can be drawn by having a vertical and a
      horizontal segment incident to each point.

      In the remaining cases, there is no such point vertically
      between $p_1$ and $p_k$.

      \begin{itemize}
        \item If $k \ge 5$ we find, similar to
          case~\ref{case:many_vertices}, two points $p, p' \in S'
          \setminus \left\{ p_1, p_k \right\}$ both below or above
          $r$ such that $p$ is the one closer to the line $y = y_r$.
          Again we set $\mu(z) = p$; if $p'$ is left of $p$ we set
          $\mu(u) = p'$ and $\mu(v) = p_k$, see
          Figure~\ref{fig:cactusbintree_case32}, and otherwise we
          symmetrically set $\mu(v) = p'$ and $\mu(u) = p_1$. Now we
          can draw the cycle without overlap such that each point is
          incident to a vertical and a horizontal segment.

        \item If $k=4$, we have $C = (u,z,v,w)$. If
          $p_2$ and $p_3$ lie both above or below $r$ we can proceed
          as in the previous case. Otherwise we know that both points
          are on different sides of $y = y_{r}$, and that $p_1$ and
					$p_4$ are both vertically between, below, or above $p_2$ and $p_3$. In
					the first case, we set $\mu(u) = p_1, \mu(v) = p_4, \mu(z) = p_2$ and
					$\mu(w) = p_3$ and create the drawing of $C$ as in
					Figure~\ref{fig:cactusbintree_case33}. As above and below are
					symmetric, the other two cases can be handled as shown in
					Figure~\ref{fig:cactusbintree_case3extra}.

          \begin{figure}[tb]
            \begin{minipage}[b]{0.48\linewidth}
              \centering
                \includegraphics{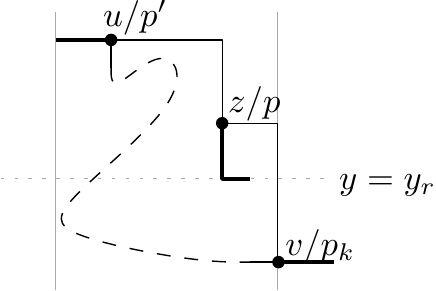}
              \caption{Drawing of $C$ with $u,v$ as neighbors of
              $z$ and $k \ge 4$.}
              \label{fig:cactusbintree_case32}
            \end{minipage}
            \hfill
            \begin{minipage}[b]{0.48\linewidth}
              \centering
                \includegraphics{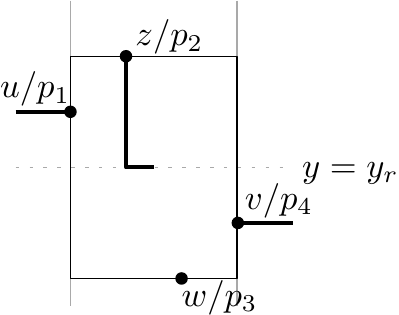}
              \caption{Drawing of $C$ with $u,v$ as neighbors of
              $z$, $k=4$, and $p_1, p_4$ vertically between $p_2$ and $p_3$.}
              \label{fig:cactusbintree_case33}
            \end{minipage}
          \end{figure}

        \item Finally, if $k = 3$, we set $\mu(u) = p_1, \mu(v) = p_3$ and
          $\mu(z) = p_2$, and simply draw as shown in
          Figure~\ref{fig:cactusbintree_case34}.

          \begin{figure}[tb]
            \begin{minipage}[b]{0.48\linewidth}
              \centering
                \includegraphics{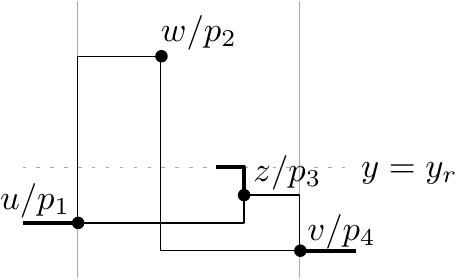}
              \caption{Drawing of $C$ with $u,v$ as neighbors of
              $z$, $k=4$, and $p_1, p_4$ vertically below $p_2$ and $p_3$.}
              \label{fig:cactusbintree_case3extra}
            \end{minipage}
						\hfill
            \begin{minipage}[b]{0.48\linewidth}
              \centering
                \includegraphics{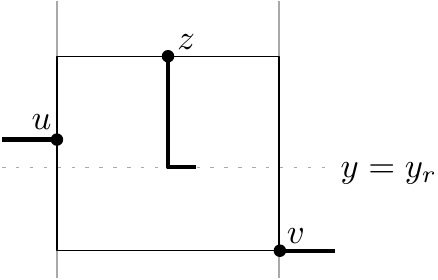}
              \caption{Drawing of $C$ with $u,v$ as neighbors of
              $z$ and $k = 3$.}
              \label{fig:cactusbintree_case34}
            \end{minipage}
          \end{figure}

      \end{itemize}

  \end{enumerate}
\end{proof}

\begin{figure}[b]
  \begin{center}
    \includegraphics{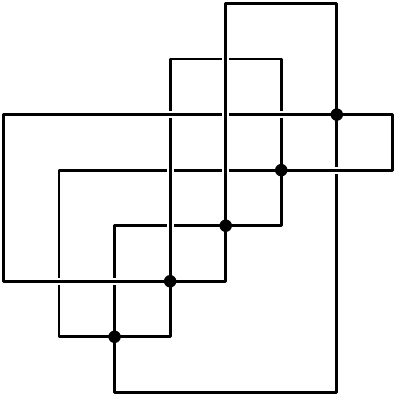}
  \end{center}
  \caption{Restricted RAC$_2$ drawing of $K_5$ on a diagonal point set.}
  \label{fig:k5_drawing_rac2}
\end{figure}
\end{document}